\documentclass[letterpaper,11pt]{article}  
\usepackage{amsfonts}
\usepackage[pdftex]{graphicx}
\usepackage{amsmath,amssymb,amsthm}
\usepackage{natbib}
\usepackage{har2nat}
\usepackage{hyperref}
\usepackage{bm}
\usepackage{caption}
\usepackage{subcaption}
\usepackage{multirow}
\usepackage{adjustbox,lipsum}
\usepackage{fancyhdr}
\usepackage{sectsty}
\usepackage{stmaryrd}
\usepackage{setspace}                      
\usepackage{booktabs}    
\usepackage{pdfsync}        
\usepackage[backgroundcolor=blue!40,linecolor=blue!40]{todonotes}
\usepackage{fancybox}
\usepackage{appendix}
\usepackage{arydshln}
\usepackage{tabularx,booktabs}

\definecolor{cornellred}{RGB}{179,27,27} 
\definecolor{cornellblue}{RGB}{00,00,170}
\definecolor{cornellgrey}{RGB}{96,94,92}

\newtheoremstyle{myplain}
  {9pt}
  {9pt}
  {\itshape}
  {\parindent}
  {\scshape}
  {:}
  {.5em}
  {}
\newtheoremstyle{mydefinition}
  {9pt}
  {9pt}
  {\itshape}
  {\parindent}
  {\scshape}
  {:}
  {.5em}
  {}
\newtheoremstyle{myremark}
  {9pt}
  {9pt}
  {}
  {\parindent}
  {\scshape}
  {:}
  {.5em}
  {}
\theoremstyle{myplain}

\newtheorem{proposition}{Proposition}
\theoremstyle{mydefinition}

\newtheorem{assumption}{Assumption}

\theoremstyle{myremark}

\newtheorem{remark}{Remark}[section]
     
\renewcommand{\cite}{\citet}

\usetikzlibrary{calc}
\def\centerarc[#1](#2)(#3:#4:#5){ \draw[#1] ($(#2)+({#5*cos(#3)},{#5*sin(#3)})$) arc (#3:#4:#5);}

\hypersetup{colorlinks=true, linkcolor=blue, citecolor=blue}                          
\numberwithin{equation}{section}

\begin{document}

\title{Bounding Infection Prevalence \\ by Bounding Selectivity and Accuracy of Tests: \\ With Application to Early COVID-19\thanks{Thanks to Chuck Manski and Francesca Molinari for stimulating discussions and for sharing their data, to Dan Sacks, Coady Wing, and  Gabriel Ziegler for feedback and literature pointers, and to Cynthia Stoye for special support in crazy times. Any and all errors are mine.}}
\date{\today}
\author{J\"{o}rg Stoye\thanks{Department of Economics, Cornell University, stoye@cornell.edu.}}

\maketitle
\begin{abstract}
I propose novel partial identification bounds on infection prevalence from information on test rate and test yield. The approach utilizes user-specified bounds on (i) test accuracy and (ii) the extent to which tests are targeted, formalized as restriction on the effect of true infection status on the odds ratio of getting tested and thereby embeddable in logit specifications. The motivating application is to the COVID-19 pandemic but the strategy may also be useful elsewhere.

Evaluated on data from the pandemic's early stage, even the weakest of the novel bounds are reasonably informative. Notably, and in contrast to speculations that were widely reported at the time, they place the infection fatality rate for Italy well above the one of influenza by mid-April.
\end{abstract}

\vfill

\pagebreak
\onehalfspacing

\section{Introduction}
\label{sec:introduction}

Prevalence of a novel infection like SARS-CoV-2 (the virus causing COVID-19 disease) is a quintessential missing data problem. Only a small subset of the population has been tested, this subset is almost certainly selective, we do not even know the accuracy of tests, and our understanding of the pandemic is vague enough so that we might not want to overly rely on heavily parameterized models. This is a natural application for partial identification analysis, i.e. the analysis of bounds on parameter values that can be inferred from imperfect data and weak but credible assumptions, without forcing statistical identifiability of a model.\footnote{See \cite{Manski03} for an early monograph and \cite{MolinariHOE} for an extensive survey.} The present paper proposes a general framework for analyzing partial identification of disease prevalence, assuming that one has partially identifying information on the \textit{selectivity} and \textit{sensitivity} of diagnostic tests.

The obvious precedent for this is \cite[][MM henceforth]{MM20}. I agree with that paper's overall thrust but propose a considerably different implementation, refining worst-case bounds by bounding test sensitivity and selectivity but not predictive value (all terms will be defined later). These restriction are readily related to other literatures, and --unlike with predictive values-- nonvacuous prior bounds on them can be asserted without implying informative prior bounds on prevalence itself. In the empirical application, bounds that only restrict the direction of selectivity are considerably more informative than the analogous bounds emphasized in MM, and yet I will argue that assumptions became more compelling. The difference matters: The novel bounds on the infection fatality rate exclude ``flu-like" values, which were the subject of speculation at the time, at a rather early stage. They become much tighter, though at the cost of reduced credibility, if one substantively restricts selectivity.

\section{The Identification Problem}

\subsection{Basic Setting and Worst-Case Bounds}

Consider first the problem of bounding prevalence of an infection in a stylized example where one has observed test rate and test yield for one population. I will call the disease COVID-19 henceforth but the ideas are obviously more general. For readability, I also follow common parlance and loosely refer to COVID-19, though I strictly speaking investigate SARS-CoV-2 infection as opposed to COVID-19 disease.  

Thus, let $C$ indicate true infection status (with $C=1$ indicating infection), $T$ test status (with $T=1$ indicating having been tested), and $R$ test result (with $R=1$ a positive test result; we observe $R$ only conditionally on $T=1$). In particular, define the testing rate $\tau:=\Pr(T=1)$ and the test yield $\gamma:=\Pr(R=1|T=1)$. These objects are directly identified from the data, and we initially assume that they are known; indeed, inference will turn out to be a secondary concern. We also maintain the assumption that (PCR-)tests have specificity (=true negative rate $\Pr(R=0|T=1,C=0)$) of $1$; thus, $\Pr(C=1|R=1)=1$. Generalizing away from this simplification would be straightforward.

Worst-case bounds on the true infection rate can then be derived from the Law of Total Probability and the logical bound of $[0,1]$ on any unknown probability. In particular, write
\begin{equation*}
\Pr(C=1) = \Pr(C=1|R=1)\Pr(R=1) + \Pr(C=1|R=0)\Pr(R=0)
\end{equation*}
and observe that $\Pr(C=1|R=0) \in [0,1]$, whereas $\Pr(R=1)=\gamma\tau$ and $\Pr(C=1|R=1)=1$ by maintained assumption. Thus, without any further assumption,
\begin{equation}
\Pr(C=1) \in [\gamma\tau, 1]. \label{eq:worst_case}
\end{equation}
These bounds go back to \cite{Manski89} in spirit and are also the starting point of MM. I next lay out novel ways to refine them.

\subsection{Introducing Bounds on Sensitivity and Selectivity of Tests}
Consider injecting prior information on \textit{test sensitivity} (i.e., true positive rate $\Pr(R=1|C=1)$) and on \textit{test selectivity} (i.e., the relation of $\Pr(T=1|C=1)$ to $\Pr(T=1|C=0)$ but not either of these two probabilites by itself). I do not claim that any of these are context-independent, much less known; hence, the prior information will itself be in the form of bounds. However, test sensitivity relates directly to a large medical literature, and test selectivity is readily related to statistical models of binary response. I next explain the approach and work out its implications.

\paragraph{Refinement: Allow for measurement error through bounding sensitivity.} Test sensitivity is the target parameter in much research on COVID-19 (and, of course, more generally). Thus, consider:
\begin{assumption} \label{as:sens}
Sensitivity of the test is bounded by
\begin{equation}
\Pr(R=1|C=1,T=1) =: \pi \in [\underline{\pi},\overline{\pi}]. \label{eq:ass_sens}
\end{equation}
\end{assumption}
The assumption takes a notational shortcut: It seemingly implies that $\pi$ is constant across true prevalence and test rates. This textbook view \citep[][chapter 2]{textbook} has been challenged \citep{LBI08}. In the specific case of COVID-19, both prevalence and testing rate might influence sensitivity through the distribution of viral load among the tested and infected; the corresponding conjecture that asymptomatic surveillance might be characterized by lower sensitivity than symptomatic surveillance has some empirical corroboration \citep{Mohammadi20,Zhang21}. However, bounds derived below do not exploit constancy of $\pi$ beyond the fact that $\pi \in [\underline{\pi},\overline{\pi}]$. So this is best thought of as a shortcut to avoid further subscripts, though users should keep the consideration in mind when specifying $(\underline{\pi},\overline{\pi})$.
 
The effect of Assumption \ref{as:sens} on prevalence bounds is easily calculated.
\begin{proposition} \label{prop:1}
Suppose Assumption \ref{as:sens} holds. Then prevalence is sharply bounded by
\begin{equation}
\rho \in [\gamma\tau/\overline{\pi}, \gamma\tau/\underline{\pi}+1-\tau].
\end{equation}
\end{proposition}
\begin{proof}
Write
\begin{eqnarray}
\Pr(C=1) = \Pr(C=1|T=1)\Pr(T=1) + \Pr(C=1|T=0)\Pr(T=0). \label{eq:sens_pf1}
\end{eqnarray}
While no informative bound on $\Pr(C=1|T=0)$ is available, we have
\begin{eqnarray*}
&& \Pr(R=1|T=1) \\
&=&\Pr(R=1|C=1,T=1)\Pr(C=1|T=1)+\underset{=0\text{ by assumption}}{\underbrace{\Pr(R=1|C=0,T=1)}}\Pr(C=0|T=1) \\
 &=&\underset{=\pi}{\underbrace{\Pr(R=1|T=1,C=1)}} \Pr(C=1|T=1),
\end{eqnarray*}
implying (in the notation introduced above) that $\Pr(C=1|T=1)=\gamma/\pi \in [\gamma/\overline{\pi},\gamma/\underline{\pi}]$. The bounds follow by substituting into \eqref{eq:sens_pf1}. 
\end{proof}
\begin{remark}
This result is easily extended to allow for specificity (=true negative rate $\Pr(R=0|C=0,T=1)$) to differ from $1$. Indeed, the bounds simply adjust prevalence in the tested population through the well-known formula ``prevalence=(yield+specificity-1)/(sensitivity+specificity-1)" and leave prevalence in the untested population unconstrained. This is not worked out to economize on notation.
\end{remark}

\paragraph{Refinement: A ``logit bound" on test selectivity.} Consider also the following:
\begin{assumption} \label{as:select}
The factor $\kappa$ in
$$ \frac{\Pr(T=1|C=1)}{1-\Pr(T=1|C=1)} = \kappa \frac{\Pr(T=1|C=0)}{1-\Pr(T=1|C=0)}$$
can be bounded as $\kappa \in [\underline{\kappa},\overline{\kappa}]$.
\end{assumption}
Assumption \ref{as:select} resembles sensitivity analysis for treatment effects in \citet{Rosenbaum02}. It bounds the relative odds ratio of being tested between true positives and true negatives. Of course, this is only one of many possible ways to constrain how targeted tests are. However, it is easily related to standard models of selection as binary response. In particular, bounding $\kappa$ in the above is equivalent to bounding it in the logit model
\begin{equation*}
\Pr(T=1|C=c)=\frac{\exp(\alpha+\kappa c)}{1+\exp(\alpha+\kappa c)}.
\end{equation*}
Logit models are well understood in econometrics and medical statistics, so this connection generates an interface to natural estimation strategies and maybe researcher intuitions about plausible parameter values.

For example, \citet{Canning20} model the age-dependent effect of COVID-19 syptoms on social distancing behavior through a logit; a similar model could in principle be used to model self-selection into symptomatic surveillance. If such a model were applied to the propensity to get tested, true infection status could be treated as hidden covariate. If one were furthermore willing to bound the coefficient on this covariate --where the bounds may depend on the values of other covariates-- then conditionally on any realization of observed covariates, one is in the setting of Assumption \ref{as:select}. In a setting of symptomatic surveillance, the propensity of noninfected subjects to get tested could, for example, relate to the age-specific frequency of influenza-like symptoms.

The selectivity factor $\kappa$ could be bounded from both above and below. For this paper's application, I will impose throughout that $\kappa \geq 1$, thus there is at least weak selection of infected subjects into testing, and I will consider values of $\kappa$ that force strict selection. Bounding selectivity from above, or also allowing for a lower bound below $1$, may be interesting in other contexts, for example, if getting tested is stigmatized or tests are targeted but not at the at-risk population.

The implications of bounding $\kappa$ are slightly more involved.
\begin{proposition} \label{prop:2}
Suppose that Assumptions \ref{as:sens} and \ref{as:select} hold. Then prevalence is sharply bounded by
\begin{equation}
\Pr(C=1) \in \left[\frac{\gamma}{\overline{\pi}} \times \frac{\overline{\pi} + (\overline{\kappa}-1)\tau(\overline{\pi}-\gamma)}{\overline{\kappa}(\overline{\pi}-\gamma)+\gamma},\frac{\gamma}{\underline{\pi}} \times \frac{\underline{\pi} + (\underline{\kappa}-1)\tau(\underline{\pi}-\gamma)}{\underline{\kappa}(\underline{\pi}-\gamma)+\gamma}\right]. 
\end{equation}
including by the corresponding limiting expressions as $\underline{\kappa} \to -\infty$ or $\overline{\kappa} \to \infty$. In particular, if $\overline{\kappa}=\infty$ as in the empirical application, we have 
\begin{equation}
Pr(C=1) \in \left[\frac{\tau \gamma}{\overline{\pi}},\frac{\gamma}{\underline{\pi}} \times \frac{\underline{\pi} + (\underline{\kappa}-1)\tau(\underline{\pi}-\gamma)}{\underline{\kappa}(\underline{\pi}-\gamma)+\gamma}\right]. \label{eq:bounds}
\end{equation}
\end{proposition}
\begin{proof}
To keep algebra transparent, introduce new notation $\tau_c\equiv\Pr(T=1|C=c)$. Recall also the notation $\rho=\Pr(C=1)$. Write
\begin{eqnarray}
\gamma = \frac{\Pr(R=1,T=1)}{\Pr(T=1)} = \frac{\rho \tau_1 \pi}{\tau} \implies \tau_1=\frac{\gamma \tau}{\rho \pi}. \label{eq:tau_1}
\end{eqnarray}
Substituting
\begin{eqnarray*}
\frac{\Pr(T=1|C=1)}{1-\Pr(T=1|C=1)} = \kappa \frac{\Pr(T=1|C=0)}{1-\Pr(T=1|C=0)} \implies \tau_0 = \frac{\tau_1}{\tau_1+\kappa(1-\tau_1)}
\end{eqnarray*}
into the accounting identity $\tau=\rho \tau_1 + (1-\rho) \tau_0$ yields
\begin{eqnarray*}
\tau = \rho \tau_1 + (1-\rho)\frac{\tau_1}{\tau_1+\kappa(1-\tau_1)}.
\end{eqnarray*}
Substituting for $\tau_1$ from \eqref{eq:tau_1} yields the following algebra:
\begin{eqnarray*}
&& \tau = \frac{\gamma\tau}{\pi} + (1-\rho)\frac{\frac{\gamma\tau}{\rho\pi}}{\frac{\gamma\tau}{\rho\pi}+\kappa\left(1-\frac{\gamma\tau}{\rho\pi}\right)} \\
&\Longleftrightarrow & \pi = \gamma+(1-\rho)\frac{\gamma\pi}{\kappa\rho\pi-(\kappa-1)\gamma\tau} \\
&\Longleftrightarrow & (\kappa\rho\pi-(\kappa-1)\gamma\tau)(\pi-\gamma) = (1-\rho)\gamma\pi \\
&\Longleftrightarrow & \rho = \frac{\gamma}{\pi} \times \frac{\pi+(\kappa-1)\tau(\pi-\gamma)}{\kappa(\pi-\gamma)+\gamma}.
\end{eqnarray*}
By taking derivatives, one can verify that the r.h. fraction in the last expression, and therefore the entire expression, decreases in both $\pi$ and $\kappa$. Bounds follow by evaluating it at $(\pi,\kappa)=(\underline{\pi},\underline{\kappa})$ respectively $(\pi,\kappa)=(\overline{\pi},\overline{\kappa})$.
\end{proof}
The bounds effectively multiply sample prevalence by an adjustment factor that reflects test selectivity. As would be expected, the implied prevalence decreases in selectivity $\kappa$ and sensitivity $\pi$. Note also (again as expected) that the adjustment factor simplifies to $\rho=\gamma/\pi$ at $\kappa=1$ (no selectivity would mean we estimate prevalence by prevalence in the tested subpopulation), to $\rho \to \tau \gamma/\pi$ as $\kappa \to \infty$ (perfect targeting means we impute zero prevalence in the untested population; compare \eqref{eq:bounds}) and also, for the record, $\rho \to 1-\tau+\tau \gamma/\pi$ as $\kappa \to 0$ (perfectly wrong targeting means we impute complete prevalence in the untested population).

\begin{remark}
While I present their cumulative impact, Assumptions \ref{prop:1} and \ref{prop:2} are in principle easily separated: The first one restricts the relation between test yield and prevalence in the tested population, the second one restricts prevalence across tested and untested populations. Readers are encouraged to ``pick and choose" and, of course, also to propose other approaches. For example, sensitivity adjustment could be combined with MM's suggestion to restrict the rate of asymptomatic infections.
\end{remark}

\subsection{Bounds on the Negative Predictive Value}
The negative predictive value $\text{NPV}=\Pr(C=0|R=0,T=1)$ is the probability that a negative test result is accurate. It is of great importance in medical decision making \citep{EB20,Manski20,Watson20}. It can be bounded as follows:
\begin{proposition} \label{prop:3}
Suppose Assumption \ref{as:sens} holds. Then sharp bounds on the NPV are given by
\begin{equation}
\text{NPV} \in \left[  \frac{1-\gamma/\overline{\pi}}{1-\gamma} , \frac{1-\gamma/\underline{\pi}}{1-\gamma}  \right]. \label{eq:eta}
\end{equation}
\end{proposition}
\begin{proof}
Also for later use, denote the NPV as $\eta := \Pr(C=0|R=0,T=1)$, then
\begin{eqnarray*}
\eta &=& \frac{\Pr(C=0,R=0|T=1)}{\Pr(C=0,R=0|T=1)+\Pr(C=1,R=0|T=1)} \\
&=& \frac{1-\Pr(C=1|T=1)}{1-\Pr(C=1|T=1)+(1-\pi)\Pr(C=1|T=1)} \\
&=& \frac{1-\gamma/\pi}{1-\gamma/\pi+(1-\pi)\gamma/\pi} = \frac{1-\gamma/\pi}{1-\gamma},
\end{eqnarray*}
where $\gamma=\pi \Pr(C=1|T=1)$ was used. This obviously decreases in $\pi$.
\end{proof}
The result is easy but will be used momentarily. Note that the expression derived in the proof has an easy intuition: The numerator is the fraction of true negatives, i.e. adjusting yield by sensitivity, whereas the denominator is the proportion of measured negatives, in the tested population. The result could again be easily generalized to also allow for specificity of less than $1$. In that case, there would also be nondegenerate bounds on the positive predictive value $\Pr(C=1|R=1,T=1)$, which equals $1$ here because of the assumption of perfect specificity.

\subsection{Comparison to Bounds that Start from NPV}\label{sec:NPV}
Assumption \ref{as:sens} contrasts with MM's strategy of inputting ex ante bounds on the NPV. In each case, bounds on the respective other quantity become an output of the model, so the direction of logical inference is reversed. Notating the input bounds as $\eta \in [\underline{\eta},\overline{\eta}]$, MM establish that
\begin{equation}
\Pr(C=1) \in [\tau (\gamma + (1-\gamma)\underline{\eta}),\gamma + (1-\gamma)\overline{\eta}]. \label{eq:npv_bounds}
\end{equation}
I will add to this the observation that in conjunction with empirical test yield, prior bounds on NPV restrict test sensitivity. Specifically, simple algebra building on Proposition \ref{prop:3} yields
\begin{equation}
\pi \in \left[\frac{\gamma}{1-(1-\gamma)\overline{\eta}},\frac{\gamma}{1-(1-\gamma)\underline{\eta}}\right]. \label{eq:bound_sens}
\end{equation}
The following are some methodological considerations as to why one might want to rather start from sensitivity and selectivity.
\begin{itemize}
\item By bounding the NPV, one necessarily directly restricts prevalence in the tested population. This is because $\Pr(C=1|T=1)=\gamma + (1-\gamma)(1-\eta)$, so the lower and upper bound on $ \Pr(C=1|T=1)$ necessarily exceed the corresponding bound on $(1-\eta)$. Since also the upper bound on overall prevalence is just the upper bound on prevalence in the tested population, the effect can be large.

In addition, \eqref{eq:bound_sens} reveals that bounds on NPV do (in conjunction with test yield, which is observable) imply bounds on sensitivity. But because these are not made explicit, an opportunity to check empirical plausbility of assumptions is missed.

To see how all of this can play out, consider MM's prior bounds of $[.6,.9]$ on the NPV. With this input, no data can move the upper bound on prevalence below $.4$ and no data --including a test yield of $0$-- can reduce the lower bound to $0$. For example, if test yield is $.1$, then prevalence in the tested population is bounded by $[.19,.46]$; the upper bound of $.46$ also applies to overall prevalence; and test sensitivity is restricted to $[.22,.53]$, far below any plausible range of values. This example is stark but not hypothetical; compare the first entry in Table 2 in MM, replicated in the first line of Table \ref{tb:1} below. About half of the upper bounds in that table are below $.5$, so it is important to understand that they cannot be below $.4$ by construction.\footnote{That bounds on NPV presuppose bounds on prevalence is clearly expressed in \cite{Manski20}, who reverses the direction of logical inference and bounds the NPV of serological tests by inputting MM's prevalence bounds which, in turn, inputted assumed (albeit for PCR tests) NPV bounds.}

In contrast, prior bounds on sensitivity do not directly restrict prevalence. This does not mean that plausible bounds on the former are necessarily independent of the latter; see the discussion after Assumption \ref{as:sens}. However, one can assert bounds that are plausible over a wide range of prevalences, and not directly restricting the latter matters. Indeed, because any value of NPV strictly below $1$ bounds prevalence away from $0$ by assumption, I submit that in the early stages of a pandemic, any such bound runs the risk of injecting ``incredible certitude" \citep{Manski11}. The above, implied bounds on sensitivity are arguably a case in point.

\item Inputting sensitivity (and possibly specificity) generates an interface with the literature on diagnostic tests because it is the focus of this literature. For a general example, see Table 1 in \cite{PS17}. With regard to COVID-19, practitioners' guides \citep{EB20,Watson20} emphasize the importance of NPV for decisions but treat sensitivity and specificity as scientific input and NPV as jointly determined by those and prevalence. The literature on diagnostic testing \citep{PCR_review,Yang20} is explicitly about sensitivity.

MM seem to disagree when they write: ``Medical experts have been cited as believing that the rate of false-negative test findings is at least $0.3$. However, it is not clear whether they have in mind one minus the NPV or one minus test sensitivity." The technical definition of false-negative rate is not in doubt, so the concern is about informal usage. This may be a valid point in general, especially  as conflation of the two corresponds to base-rate neglect, but it did not occur to me with regard to the literature on COVID-19.\footnote{The footnote accompanying the cited sentence links to a news piece that attributes an estimated false-negative rate of $.3$ to \citet{Yang20}. While the news piece has vague language, \citet{Yang20} unambiguously estimates one minus sensitivity.}

\item Asserting bounds on NPV without taking targeting of tests into account may ignore constraining information that could lead to tighter bounds. Specifically, relatively low values of the NPV (i.e., a large fraction of negative test results being false) will be more plausible if one believes the test to be efficiently targeted. But in that same case one would conclude that the constraint $\Pr(C=1|T=1) \geq \Pr(C=1|T=0)$ is far from binding. Therefore, the degree of targeting informally enters NPV-based bounds twice, in different directions, but derivation of \eqref{eq:npv_bounds} does not force the value to be the same in both appearances. Assumption \ref{as:select} is intended to allow for targeting of tests to affect bounds in a disciplined manner. 
\end{itemize}

\section{Empirical Application}\label{sec:empirics}
These bounds are mainly designed to process the information that is available early in a pandemic. For this reason and to highlight some important differences, I illustrate them on MM's data, i.e. daily counts of tests, test results, and fatalities for Illinois, New York, and Italy in March and April, and extend them only to early analysis of subsequent hot spots. For credible application to prevalence data at a much later stage of a pandemic, one would have to take into account multiple testing and other factors.  

The leftmost columns of Tables \ref{tb:1}-\ref{tb:2} present bounds that set $\kappa \in [1,\infty)$, that is, they only restrict the direction of selectivity. I also restrict sensitivity to be in $[.7,.95]$. This is the same sensitivity interval that was used by \cite{Frazier20} in the analysis on which Cornell's Fall reopening plans were based and was supported by the medical literature at the time.\footnote{\cite{UCSF20} base medical advice on a point estimate of $.8$. \cite{Watson20} give $.7$ as ``lower end of current estimates from systematic reviews." \cite{Frazier20} use a preferred point estimate of $.9$. Recall in particular that the data analyzed here were overwhelmingly generated by testing of symptomatic subjects.}   
For comparison, the third column of Table \ref{tb:1} presents bounds that restrict NPV to be in $[.6,.9]$ as well as the direction of selectivity. These replicate MM's Table 2.\footnote{MM's results were independently replicated from their original data. MATLAB code generating all tables is available from the author. To keep the presentation succinct, the tables only show one day per week of data and the last day.\\MM refine these bounds by imposing time monotonicity; that is, prevalence (and therefore both bounds on it) cannot decrease over time. I agree with that restriction and can provide tables that implement it. It is dropped here solely because those tables have many identical rows.} Table \ref{tb:2} extends these same bounds to new data.\footnote{The reader should keep in mind that MM might not have asserted the same bounds on NPV there, so the bounds may not be what they would have proposed. On the other hand, it is a feature that one can (in this author's view) comfortably impose the same bounds on $\pi$ across these contexts.} To facilitate further comparisons, the tables also illustrate the bounds on NPV implied by bounds on sensitivity (second column, corresponding to Proposition \ref{prop:3}) and the converse bounds if one starts from restricting NPV (last column, corresponding to \eqref{eq:bound_sens}).

\begin{table}
\adjustbox{max width=\textwidth, max totalheight=\textheight}{
\begin{tabular}{cccccc} && \multicolumn{2}{c}{\textbf{New Bounds ($\bm{\overline{\kappa}=\infty}$)}} & \multicolumn{2}{c}{\textbf{NPV-based Bounds}} \\
 & \textbf{Date} & \textbf{Prevalence} & \textbf{NPV}& \textbf{Prevalence} & \textbf{Sensitivity} \\[1.5ex]
\multirow{7}{*}{\textbf{Illinois}} & 3/16 & [0.000,0.131] & [0.957,0.995] & [0.000,0.455] & [0.202,0.503]\\
& 3/23 & [0.000,0.186] & [0.936,0.992] & [0.000,0.478] & [0.272,0.599]\\
& 3/30 & [0.000,0.237] & [0.915,0.990] & [0.000,0.500] & [0.332,0.666]\\
& 4/06 & [0.001,0.279] & [0.896,0.987] & [0.001,0.517] & [0.377,0.708]\\
& 4/13 & [0.002,0.297] & [0.887,0.986] & [0.002,0.525] & [0.396,0.724]\\
& 4/20 & [0.003,0.303] & [0.885,0.986] & [0.003,0.527] & [0.402,0.729]\\
& 4/24 & [0.003,0.299] & [0.887,0.986] & [0.004,0.525] & [0.398,0.725] \\ [1.5ex]
\multirow{7}{*}{\textbf{New York}} & 3/16 & [0.000,0.191] & [0.934,0.992] & [0.000,0.480] & [0.279,0.607]\\
& 3/23 & [0.001,0.400] & [0.833,0.980] & [0.002,0.568] & [0.493,0.795]\\
& 3/30 & [0.004,0.527] & [0.749,0.969] & [0.005,0.621] & [0.594,0.854]\\
& 4/06 & [0.007,0.583] & [0.705,0.964] & [0.008,0.645] & [0.633,0.873]\\
& 4/13 & [0.011,0.579] & [0.708,0.964] & [0.012,0.643] & [0.630,0.872]\\
& 4/20 & [0.013,0.554] & [0.728,0.967] & [0.015,0.633] & [0.613,0.864]\\
& 4/24 & [0.015,0.519] & [0.756,0.970] & [0.017,0.618] & [0.588,0.851] \\[1.5ex]
\multirow{7}{*}{\textbf{Italy}} & 3/16 & [0.000,0.290] & [0.891,0.987] & [0.001,0.522] & [0.389,0.718]\\
& 3/23 & [0.001,0.331] & [0.871,0.984] & [0.002,0.539] & [0.430,0.751]\\
& 3/30 & [0.002,0.304] & [0.884,0.986] & [0.002,0.528] & [0.404,0.730]\\
& 4/06 & [0.002,0.263] & [0.903,0.988] & [0.003,0.510] & [0.361,0.693]\\
& 4/13 & [0.003,0.217] & [0.923,0.991] & [0.004,0.491] & [0.309,0.642]\\
& 4/20 & [0.003,0.186] & [0.936,0.992] & [0.005,0.478] & [0.272,0.599]\\
& 4/24 & [0.003,0.169] & [0.943,0.993] & [0.006,0.471] & [0.251,0.572]
 \end{tabular}
}
\caption{Bounds from Proposition \ref{prop:2} ($\overline{\kappa}=\infty$) and \ref{prop:3} assuming that sensitivity is bounded by $[.7,.95]$. For comparison, the right-hand columns are bounds on prevalence and sensitivity are generated by bounding negative predictive value to be in $[.6,.9]$.} \label{tb:1}
\vspace{1cm}

\adjustbox{max width=\textwidth, max totalheight=\textheight}{
\begin{tabular}{cccccc} && \multicolumn{2}{c}{\textbf{New Bounds ($\bm{\overline{\kappa}=\infty}$)}} & \multicolumn{2}{c}{\textbf{NPV-based Bounds}} \\
 & \textbf{Date} & \textbf{Prevalence} & \textbf{NPV}& \textbf{Prevalence} & \textbf{Sensitivity} \\[1.5ex]
\textbf{Arizona} & 8/13 & [0.028,0.258] & [0.906,0.988] & [0.038,0.508] & [0.355,0.687]\\
\textbf{California}& 8/13 & [0.016,0.090] & [0.971,0.996] & [0.037,0.438] & [0.143,0.401]\\
\textbf{Florida} & 8/13& [0.027,0.193] & [0.933,0.992] & [0.043,0.481] & [0.281,0.610]\\
\textbf{Texas}& 8/13 & [0.019,0.173] & [0.941,0.993] & [0.031,0.473] & [0.257,0.580]
 \end{tabular}
}
\caption{Extension of Table \ref{tb:1} to U.S. pandemic hot spots later in 2020.} \label{tb:2}
\end{table} 

The new upper bounds are considerably more restrictive and the lower bounds are slightly less so; in sum, all bounds move down. For the first day of data in Illinois, the comparison is between upper bounds of $13\%$ versus $46\%$. The effect on lower bounds is less pronounced but is not negligible in relative terms; for example, for Illinois on $3/23$, rounding error obscures that the NPV-based lower bound is $1.5$ times the sensitivity-based one. The difference is also substantively meaningful. The new bounds would rather clearly have ruled out speculation of saturation being ``around the corner" at the time. Consider also the implied bounds on the infection fatality rate (i.e., fatalities divided by true infections; IFR henceforth). The most informative NPV-based lower bounds (i.e., evaluated on 4/24) equal $.0003$ for Illinois, $.0013$ for New York, and $.0010$ for Italy. This is close to ``flu-like" numbers that were the subject of speculation at the time, and so it might have appeared that credible partial identification analysis did not exclude such speculations. Yet it did: For the same data, the bounds from Proposition \ref{prop:3} are $.0005$, $.0016$, and $.0026$; for Italy, the lower bound is above $.001$ starting on 3/29. In places where the data admittedly spoke very loudly, these numbers would have cast strong doubts on ``just the flu" conjectures in real time.

Tighter bounds are an unambiguous improvement only if assumptions did not become less credible. I would indeed argue that credibility might, if anything, have improved. A symptom of this is that the NPV-based bounds frequently imply sensitivity below $.7$, whereas the sensitivity-based bounds imply NPV mostly close to, and frequently above, $.9$. The former number seems out of step with expert opinion, including at the time, whereas the latter one would not have raised any eyebrows.\footnote{As part of a recent partial identification analysis, \cite{Sacks20} provide an empirically informed NPV estimate for Indiana of $.995$. This comes with caveats: It corresponds to obviously lower prevalence than in the data considered here, so that MM would presumably have inputted different NPV bounds; also, it operationalizes NPV as test-retest validity. \cite{UCSF20} gives NPV as $.972$ for symptomatic and $.998$ for asymptomatic cases in the Bay Area, though using the sort of point-identifying assumptions that we seek to avoid here.} Also, Table \ref{tb:1} reveals that according to NPV-based bounds, sensitivity increased in New York (the bounds fail to overlap). This might have happened, but forcing it by assumption is arguably against the spirit of weak and credible partial identification assumptions. I would argue that this illustrates methodological qualms, notably the first bullet in Section \ref{sec:NPV}. As discussed there, the implausibilities could be avoided by relaxing prior bounds on NPV to $[.6,1]$. However, the upper bounds would then just be worst-case bounds, all bounds would be wider than the new ones in the present data, and the intriguing feature that MM are able to exclude the crude case fatality rate (i.e., observed fatalities divided by observed cases) as true IFR would be lost. 

Table \ref{tb:2} repeats the exercise for data from subsequent hot spots of the pandemic.\footnote{Test counts and results were retrieved from the COVID tracking project. State populations are U.S. Census estimates for 7/1/19.} I deliberately restrict attention to states with high test yield because it seems that MM calibrated their input bounds to such places, and also to states that were in their first wave because that is what the bounds are designed for. NPV-based bounds continue to allow for very high prevalence but also force test sensitivity to be relatively low. Sensitivity-based upper bounds are again at most half -- often much less -- than their NPV-based counterparts, and other implications of respective bounds are roughly as before.

\begin{table} 
\adjustbox{max width=\textwidth, max totalheight=\textheight}{
\begin{tabular}{cccccccc}& \textbf{Date}& \textbf{Lower Bound} & \multicolumn{5}{c}{\textbf{Upper Bound with}} \\  & & & $\bm{\kappa\geq 5}$ & $\bm{\kappa\geq 3}$ & $\bm{\kappa\geq 2}$ & $\bm{\kappa\geq 1.5}$ & $\bm{\kappa\geq 1}$ \\[1.5ex]
\multirow{7}{*}{\textbf{Illinois}} & 3/16 & 0.000 & 0.036 & 0.048 & 0.0700 & 0.092 & 0.131\\
& 3/23 & 0.000 & 0.054 & 0.071 & 0.102 & 0.132 & 0.186\\
& 3/30 & 0.000 & 0.072 & 0.094 & 0.135 & 0.172 & 0.237\\
& 4/06 & 0.001 & 0.089 & 0.115 & 0.162 & 0.205 & 0.279\\
& 4/13 & 0.002 & 0.097 & 0.125 & 0.175 & 0.220 & 0.297\\
& 4/20 & 0.003 & 0.100 & 0.129 & 0.180 & 0.226 & 0.303\\
& 4/24 & 0.003 & 0.099 & 0.127 & 0.177 & 0.222 & 0.299 \\[1.5ex] 
\multirow{7}{*}{\textbf{New York}}& 3/16 & 0.000 & 0.056 & 0.073 & 0.106 & 0.136 & 0.191\\
& 3/23 & 0.001 & 0.144 & 0.183 & 0.251 & 0.308 & 0.400\\
& 3/30 & 0.004 & 0.221 & 0.274 & 0.360 & 0.427 & 0.527\\
& 4/06 & 0.007 & 0.264 & 0.322 & 0.414 & 0.484 & 0.583\\
& 4/13 & 0.011 & 0.264 & 0.321 & 0.411 & 0.480 & 0.579\\
& 4/20 & 0.013 & 0.248 & 0.302 & 0.389 & 0.457 & 0.554\\
& 4/24 & 0.015 & 0.224 & 0.274 & 0.357 & 0.422 & 0.519\\[1.5ex] 
\multirow{7}{*}{\textbf{Italy}} & 3/16 & 0.000 & 0.093 & 0.120 & 0.170 & 0.214 & 0.290\\
& 3/23 & 0.001 & 0.111 & 0.143 & 0.199 & 0.249 & 0.331\\
& 3/30 & 0.002 & 0.100 & 0.129 & 0.180 & 0.226 & 0.304\\
& 4/06 & 0.002 & 0.084 & 0.108 & 0.153 & 0.193 & 0.263\\
& 4/13 & 0.003 & 0.067 & 0.087 & 0.123 & 0.157 & 0.217\\
& 4/20 & 0.003 & 0.057 & 0.073 & 0.104 & 0.133 & 0.186\\
& 4/24 & 0.003 & 0.051 & 0.066 & 0.094 & 0.120 & 0.169
\end{tabular}
 }
\caption{Change in upper bounds from Table \ref{tb:1} as test selectivity is increasingly restricted.} \label{tb:select}
\vspace{1cm}

\adjustbox{max width=\textwidth, max totalheight=\textheight}{
\begin{tabular}{cccccccc}& \textbf{Date}& \textbf{Lower Bound} & \multicolumn{5}{c}{\textbf{Upper Bound with}} \\  & & & $\bm{\kappa\geq 5}$ & $\bm{\kappa\geq 3}$ & $\bm{\kappa\geq 2}$ & $\bm{\kappa\geq 1.5}$ & $\bm{\kappa\geq 1}$ \\[1.5ex]
\textbf{Arizona} & 8/13 & 0.028 & 0.093 & 0.126 & 0.164 & 0.198 & 0.258\\
\textbf{California} &8/13 & 0.016 & 0.036 & 0.046 & 0.057 & 0.068 & 0.090\\
\textbf{Florida} &8/13 & 0.027 & 0.074 & 0.097 & 0.123 & 0.148 & 0.193\\
\textbf{Texas} &8/13 & 0.019 & 0.060 & 0.081 & 0.106 & 0.130 & 0.173   \end{tabular}
 }
\caption{Change in upper bounds from Table \ref{tb:2} as test selectivity is increasingly restricted.} \label{tb:select_current}
\end{table}

Table \ref{tb:select} shows the effect of increasingly restricting test selectivity through Assumption \ref{as:select}. Reading it from right to left (meant to evoke ``outside in"), it starts with $\underline{\kappa}=1$, i.e. test monotonicity, and progresses through arguably weak restrictions up to $\underline{\kappa}=5$, which is restrictive and may be more in the spirit of a sensitivity parameter. Upper bounds respond strongly. This is reflected in the implied lower bounds on the IFR; for Italy, these increase to (in order) $.0036$, $.0046$, $.0065$, and $.0102$. Of course, these numbers should not be compared to MM's Table 2; to the contrary, MM reach similar conclusions when restricting the proportion of asymptomatic infections. The same exercise but for current hot spots is displayed in Table \ref{tb:select_current}.

The lower bounds are driven by the possibility that all true positives got tested. While this paper focuses on upper bounds, one could also use Assumption \ref{as:select} to refine lower bounds away from that scenario. For the record, restricting $\kappa\leq 100$ [$\kappa\leq 10$] would refine the (last period) lower bound on prevalence for Italy from $.0034$ to $.0047$ [$.0325$]. The upper bound on IFR would be refined from $.1278$ to $.0908$ [$.0132$], a not completely vacuous restriction.

Inference on these bounds has at least two nonstandard aspects: As MM point out, one might think of states and regions as populations of interest rather than samples from meta-populations; at the very least, one might be interested in inference conditionally on the realized population. In that sense, conventional sampling theory might not apply (and is omitted in MM for that reason). However, whether a given subject is tested and the result of that test are realizations of well-defined random variables, opening a clear avenue for statistical inference. Separately, such inference might be complicated if the variables interact (e.g., the marginal tested subject is asymptomatic) and will also involve small probabilities, so that Central Limit Theorem-based approximations (including many forms of the bootstrap) would not apply. Questions like this inform an exciting strand of current research \citep{Rothe20, Toulis20}. However, they are orthogonal to the thrust of this paper and also less salient in the application because sample sizes are so large, and identified intervals so long, that estimation uncertainty is dwarfed by identification issues.\footnote{For a paper-and-pencil computation, approximate the distribution of $(\hat{\tau},\hat{\gamma})$ as independently normal. Then the distribution of estimated bounds follows easily. Because these bounds are ordered by construction, inference would be a direct application of \cite{IM04} and in  particular \citet[][Lemma 3]{Stoye09} and would practically amount to intersecting one-sided confidence intervals. Simple calculations show that the standard errors on $(\hat{\gamma},\hat{\tau})$ would be at least two orders of magnitude smaller than the estimators; consequently, the difference between these confidence intervals and the estimated bounds is at most comparable to the tables' rounding errors.}

\section{Conclusion}
This paper proposes new methods to bound prevalence of a disease from partially identifying data and assumptions. It is to some extent intended as ``think piece" to alert researchers to the possibly fruitful application of partial identification methods. I have no doubt that domain knowledge may inform further, and better, iterations.

Options for refinement and subsequent analysis abound. Users who are comfortable with injecting more prior information may also refine bounds by placing priors on unidentified parameters \citep{Bollinger20}. Users who find the present paper's simple input bounds too coarse but do not want to commit to a prior could also interpolate between these approaches with sets of priors, i.e. in the spirit of Robust Bayesian Analysis \citep{robust_Bayes}. The analysis can also be used as input for decision recommendations. Here, a partially identifying analysis will typically only partially identify optimal actions, and a point valued decision may require to commit to a specific optimality criterion under ambiguity \citep{Manski00,Stoye12}. 

The conceptual innovation is to think of test accuracy as (unknown, not necessarily constant, and possibly not even identifiable) technological parameter and of test selectivity as something that econometric or epidemiological models can speak to. Bounds are therefore constructed with these as starting points, deriving bounds on predictive values by implication and not imposing any prior bound on prevalence in the tested population. In the empirical application, it turns out that some of the more audacious speculations floated at the time contradicted credible partial identification analysis even then. This illustrates the potential utility of such analysis in early stages of a pandemic.

That said, many of this paper's simplifications are a stretch in the current, more advanced stage of the pandemic. For example, one should distinguish between current and past infection and take multiple testing into account. I leave further exploration of such extensions to future work, but would recommend to use restrictions on test sensitivity as primitive of the analysis. Once again, my main hope is to stimulate further research on partial identification in, or in collaboration with, epidemiology.

\bibliography{covid}

\end{document}